\documentclass[sigconf]{acmart}

\AtBeginDocument{%
  \providecommand\BibTeX{{%
    \normalfont B\kern-0.5em{\scshape i\kern-0.25em b}\kern-0.8em\TeX}}}

\setcopyright{acmcopyright}
\copyrightyear{2018}
\acmYear{2018}
\acmDOI{XXXXXXX.XXXXXXX}

\acmConference[Conference acronym 'XX]{Make sure to enter the correct
  conference title from your rights confirmation emai}{June 03--05,
  2018}{Woodstock, NY}
\acmPrice{15.00}
\acmISBN{978-1-4503-XXXX-X/18/06}

\usepackage{xspace}

\usepackage{graphicx}
\usepackage{xcolor}
\usepackage{adjustbox}
\usepackage{caption}
\usepackage{subcaption}
\usepackage{microtype}
\usepackage[ruled,linesnumbered]{algorithm2e}
\usepackage{mathtools}
\DeclareUnicodeCharacter{2113}{\ell}

\newcommand{\claire}[1]{{\color{black} #1}}

\newcommand{\michel}[1]{{\color{black} #1}}

\title{The SpaceSaving± Family of Algorithms for Data Streams with Bounded Deletions}

\author{Fuheng Zhao}
\affiliation{%
  \institution{UC Santa Barbara}
  \country{}
}

\email{fuheng\_zhao@ucsb.edu}

\author{Divyakant Agrawal}
\affiliation{%
  \institution{UC Santa Barbara}
  \country{}
}
\email{agrawal@cs.ucsb.edu}

\author{Amr El Abbadi}
\affiliation{%
  \institution{UC Santa Barbara}
  \country{}
}
\email{amr@cs.ucsb.edu}

\author{Claire Mathieu}
\affiliation{%
  \institution{CNRS and IRIF}
  \country{}
}
\email{clairemmathieu@gmail.com}

\author{Ahmed Metwally}
\affiliation{%
  \institution{Uber Inc.}
  \country{}
}
\email{ametwally@gmail.com}

\author{Michel de Rougemont}
\affiliation{%
  \institution{University Paris II and IRIF}
  \country{}
}
\email{m.derougemont@gmail.com}

\begin{document}

\begin{abstract}
In this paper, we present an advanced analysis of near optimal algorithms that use limited space to solve the frequency estimation, heavy hitters, frequent items, and top-k approximation in the bounded deletion model. We define the family of SpaceSaving± algorithms and explain why the original SpaceSaving± algorithm only works when insertions and deletions are not interleaved. Next, we propose the new Double SpaceSaving±, Unbiased Double SpaceSaving±, and Integrated SpaceSaving± and prove their correctness. The three proposed algorithms represent different trade-offs, in which Double SpaceSaving± can be extended to provide unbiased estimations while Integrated SpaceSaving± uses less space. Since data streams are often skewed, we present an improved analysis of these algorithms and show that errors do not depend on the hot items. We also demonstrate how to achieve relative error guarantees under mild assumptions. Moreover, we establish that the important mergeability property is satisfied by all three algorithms, which is essential for running the algorithms in distributed settings.
\end{abstract}



 %TODO mandatory; please add comma-separated list of keywords




\maketitle

\section{Introduction}
Streaming algorithms~\cite{cormode2020small, mining2006data} aim to process a \claire{stream of data} in a single pass with small \claire{space}. In particular, \claire{for a stream of updates of items,} these algorithms often only store a compact summary which \claire{uses} space \claire{at most} logarithmic \claire{in} the data stream length or \claire{in} the cardinality of the input domain. The streaming paradigm provides essential analysis and statistical measurements with strong accuracy guarantees for many big data applications~\cite{agrawal2011challenges}. For instance, Hyperloglog~\cite{flajolet2007hyperloglog} \claire{answers} queries about cardinality; Bloom Filters~\cite{bloom1970space} answer membership queries; KLL~\cite{karnin2016optimal, ivkin2019streaming, zhao2021kll} gives accurate quantile approximations such as finding the median. These data summaries (sketches) are now the cornerstones for many real-world systems and applications including network monitoring~\cite{yang2018elastic, zhao2023panakos, xu2023mimosketch}, storage engine~\cite{dayan2017monkey, zhao2023autumn}, data mining~\cite{wang2019memory, li2020wavingsketch}, federated learning~\cite{rothchild2020fetchsgd, jiang2018sketchml}, privacy-preserving analytic~\cite{chen2023communication, zhao2022differentially,lebeda2023better, pagh2022improved, wang2024dpsw}, vector search~\cite{broder1997resemblance, bruch2023approximate}, and many more. We refer the reader to a recent paper~\cite{cormode2023gems} for a more thorough discussion on the applications.

In this paper, we focus on the frequency estimation and heavy hitters problems with the presence of deletions. Given a data stream $\sigma$ from some universe $U$, the \textbf{Frequency Estimation} problem \claire{constructs a summary providing,} for every item $x$ in $U$,  an estimate $\hat{f}(x)$ of the true frequency $f(x)$. Since resources are limited (small memory footprint), the data summary cannot track every item and hence there is an inherent trade-off between the space \claire{used} and \claire{the accuracy of the} estimate. In the closely related approximate \textbf{Heavy Hitters}  (\textbf{Frequent Items}) problem, there is an additional input frequency threshold $T$, and \claire{the goal is to} identify all items with frequency larger than $T$. \claire{This applies to settings where} data scientists and system operators are interested in knowing what items are hot (appear very often). Frequency estimation, heavy hitters and approximate top-k have found many applications in big data systems~\cite{bailis2017macrobase,zakhary2021cache,liu2023treesensing, chiosa2021skt}.

\begin{table*}[tbph]
    \centering
    \begin{tabular}{c|c|c|c|c}
    \hline
        {\bf Algorithm} & {\bf Space} & {\bf Model} & {\bf Error Bound}  & {\bf Note}  \\ 
        \hline
        \hline
        
        SpaceSaving $\&$ MG  & $O(\frac{1}{\epsilon})$ & Insertion-Only & $|f_i-\hat{f}_i| \leq \epsilon F_1$ & Lemma~\ref{spacesaving error bound} \\

        ~\cite{metwally2005efficient, misra1982finding, berinde2010space, RC23} &$O(\frac{k}{\epsilon})$ & & $|f_i-\hat{f}_i| \leq \frac{\epsilon}{k} F_1^{res(k)}$& Lemma~\ref{spacesaving residual error bound}\\

         & $O(\frac{k}{\epsilon} \frac{2-\gamma}{2(\gamma - 1)})$ & ($\gamma$-decreasing) & $|f_i-\hat{f}_i| \leq \epsilon f_{i}, \forall i \leq k$ & Lemma~\ref{spacesaving relative error bound} \\
        \hline
        
        Count-Min~\cite{cormode2005improved} & $O(\frac{k}{\epsilon}\log n)$ & Turnstile & $|f_i-\hat{f}_i| \leq \frac{\epsilon}{k} \epsilon F_1^{res(k)}$ & Never underestimate\\ \hline
        
        CountSketch~\cite{charikar2002finding} &  $O(\frac{k}{\epsilon}\log n)$ & Turnstile & $(f_i - \hat{f}_i)^2 \leq \frac{\epsilon}{k} F_2^{res(k)}$ & unbiased\\ \hline
        
        DoubleSS$\pm$ \& IntegratedSS$\pm$ & $O(\frac{\alpha}{\epsilon})$ & Bounded Deletion & $|f_i-\hat{f}_i| \leq \epsilon F_1$ & Theorem~\ref{dss theorem: frequency estimation} \&~\ref{theorem: frequency estimation} \\ 

         Unbiased DoubleSS$\pm$ & $O(\frac{\alpha}{\epsilon})$ & & unbiased \& Var < $\epsilon^2F_1^2$ & Theorem~\ref{dss theorem: unbiased} \\ 
        
        DoubleSS$\pm$ \& IntegratedSS$\pm$ &  $O(\frac{\alpha k}{\epsilon})$ &  & $|f_i-\hat{f}_i| \leq \frac{\epsilon}{k} F_{1.\alpha}^{res(k)}$ & Theorem~\ref{dss+- residual error} \&~\ref{ss+- residual error} \\ 
        
        DoubleSS$\pm$ \& IntegratedSS$\pm$ &  $O(\frac{\alpha k}{\epsilon} \frac{(2-\gamma)k^{\beta}}{2(\gamma-1) 2^{log_{\gamma}k}})$ & ($\gamma$-decreasing and $\beta$-zipf) & $|f_i-\hat{f}_i| \leq \epsilon f_i, \forall i \leq k$ & Theorem~\ref{dss+- relative error} \&~\ref{ss+- relative error}\\
        \hline

    \end{tabular}
    \caption{Comparison between different frequency estimation and approximate heavy hitters algorithms.}
    \label{tab:algorithm-comparison}
\end{table*}

\subsection{Related Work}
There is a large body of algorithms proposed to solve frequency estimation and heavy hitters problems~\cite{misra1982finding, metwally2005efficient, charikar2002finding, cormode2005improved, karp2003simple, demaine2002frequency, braverman2017bptree, jayaram2018data, zhao2021spacesaving, zhao2022differentially, manku2002approximate}. The existing algorithms can be classified into three models: insertion-only, turnstile, and bounded-deletion. In the insertion-only model, all arriving stream items are insertions \claire{and the frequency $f(x)$ of $x$ is the number of times $x$ has been inserted}. SpaceSaving~\cite{metwally2005efficient} and MG~\cite{misra1982finding} are two popular algorithms that operate in the insertion-only model; they provide strong deterministic guarantees with minimal space. Moreover, since these algorithms are deterministic, they are inherently adversarial robust~\cite{hasidim2020adversarially, ben2022framework, cohen2022robustness}. In the turnstile model, arriving stream items can be either insertions or deletions \claire{and the (net) frequency $f(x)$ of $x$ is the difference between the number of times $x$ has been inserted and deleted, with the proviso that} no item's frequency is negative. Count-Min~\cite{cormode2005improved} and CountSketch~\cite{charikar2002finding} are two popular algorithms that operate in the turnstile model. Supporting deletions is a harder task and as a result Count-Min and Count-Sketch require more space with poly-log dependence on the input domain size to offer strong error bounds with high probability~\cite{alon1996space}. The bounded deletion model~\cite{jayaram2018data, zhao2021kll, zhao2021spacesaving} assumes the number of deletions are bounded by \claire{a fraction of} the number of insertions. For instance, in standard testings, if students are only allowed to submit one regrade request on their tests, then the number of deletions is at most half of the number of insertions in which updating a student's grade is a deletion followed by an insertion. 

Recently, SpaceSaving$\pm$~\cite{zhao2021spacesaving} proposed 
\claire{an algorithm extending the original SpaceSaving~\cite{metwally2005efficient} to handle bounded deletions efficiently, (implicitly) assuming~\cite{zhao2021spacesaving} that the data stream has no interleaving between insertions and deletions: all insertions precede all deletions~\cite{metwally2005efficient}~\cite{zhao2023errata}.} 
In addition, Berinde et al.~\cite{berinde2010space} showcased that data summaries in the insertion-only model can achieve an even stronger error bound using residual error analysis. The difference among these algorithms and their corresponding error guarantees is listed in Table~\ref{tab:algorithm-comparison}. 
As a result, in this work, we aim to investigate the SpaceSaving$\pm$ family of algorithms to support interleavings between insertions and deletions and provide a tighter error bound analysis.

% (see Lemma~\ref{spacesaving+/- guarantee}). 

\subsection{Our Contributions}
In this paper, we present a detailed analysis of the SpaceSaving$\pm$ family of algorithms with bounded deletions. In particular, we propose three new algorithms: Double SpaceSaving$\pm$, Unbiased Double SpaceSaving$\pm$, and Integrated SpaceSaving$\pm$, which operate in the general bounded deletion model (with interleavings of insertions and deletions) and provide strong theoretical guarantees with small space, while having distinct characteristics. In addition, we introduce the residual error bound guarantee in the bounded deletion model and prove that the proposed algorithms satisfy the stronger residual error guarantee. Moreover, we demonstrate that the proposed algorithms also achieve the relative error bound~\cite{RC23,cormode2021relative} under mild assumptions, assuming skewness in the data stream. Finally, we provide a merge algorithm for our proposed algorithms and prove that all proposed algorithms are mergeable in the bounded deletion model. The mergeability is critical for distributed applications~\cite{agarwal2012mergeable}.

\section{Background}
\subsection{Preliminaries}
{Consider a data stream consisting of a sequence of $N$ operations (insertions or deletions) on items drawn from a universe $U$ of size $|U|=n$. We denote the stream by $\sigma = \{(item_{t},op_t)\}_{t=\{1,2...N\}}$ where $item_t\in U$ is an element of the universe and $op_t\in \{insertion,deletion\}$. 
Let $I$ denote the total number of insertions and $D$ denote the total number of deletions in the stream, so that $I+D=N$. In the {\em bounded deletion model}~\cite{jayaram2018data,zhao2021kll}, it is assumed that, for some parameter $\alpha \geq 1$, the stream is constrained to satisfy  $D \leq (1-1/\alpha)I$. 
Observe that when $\alpha=1$, there can be no deletions, and so the model includes the insertion-only model as a special case. Let $f$ denote the frequency function: the frequency of an item $x$ is $f(x) = I(x) - D(x) $ where $I(x)$ denotes the number of insertions of $x$ in the stream and $D(x)$ denotes the number of deletions of $x$ in the stream. In this paper, we use a {\em summary} data structure from which we infer an estimation of the frequency function. We use $\hat{f}(x)$ to denote the estimated frequency of $x$.}

\iffalse
Given a data stream $\sigma = \{item_{t}\}_{t=\{1,2...N\}}$ of length $N$ where items are drawn from universe $U$, $|U|$ represents the size of the universe and $n$ represents the cardinality of $\sigma$. $\sigma$ contains $I$ insertions and $D$ deletions in which $D \leq (1-1/\alpha)I$ in the bounded deletion model and when $\alpha$ equals to one, then it recovers the insertion-only model. The frequency of an items $x$ is $f(x) = I(x) - D(x) = \sum_{t=1}^{N} \pi(item_{t}, x)$ in which $\pi(item, x)$ returns 1 if the inserted item is $x$, 0 if the item is not x, and -1 if the deleted item is $x$. We use $\hat{f}(x)$ to represent the estimated frequency of $x$.
\fi

{
Let $\{f_i\}_{1}^{n}$ denote the frequencies of the $n$ items, sorted in non-increasing order, so that $f_1 \geq f_2,...\geq f_n$. Similarly, sorting $\{ I(x):x\in U\}$ in non-increasing order leads to $\{I_i\}_{1}^{n}$ where $I_1 \geq I_2,...\geq I_n$, and sorting $\{D(x): x\in U\}$  in non-increasing order leads to $\{D_i\}_{1}^{n}$ where $D_1 \geq D_2,... \geq D_n$.
} 
\iffalse 
In addition, we define three vectors $\{f_i\}_{1}^{n}$, $\{I_i\}_{1}^{n}$, and $\{D_i\}_{1}^{n}$ corresponded to frequency, insertion, and deletions. We assume, without loss of generality, that these three vectors have items indexed in decreasing order based on frequency count, insert count, and delete count such that , $I_1 \geq I_2,...\geq I_n$, and $D_1 \geq D_2,... \geq D_n$. 
\fi  
Notice that, for the same $i$, the item corresponding to $f_i$ may differ from the item corresponding to $I_i$ and  from the item corresponding to $D_i$. For the special case of the insertion-only model, the frequency vector and insertion vector are identical. Moreover, we use $F_1$ to denote the total number of items resulting from the stream operations, such that $F_1 = I-D =  \sum_{i=1}^{n} f_i$. In this paper, we focus on
the unit updates model (an operation only inserts or deletes one count of an item) and assume that at any point of time, no item frequency is below zero (i.e, items cannot be deleted if
they were not previously inserted).

\begin{definition} [Deterministic Frequency Estimation Problem] {Given a stream of operations (insertions or deletions) of items from a universe $U$, the Frequency Estimation Problem asks for an (implicitly defined) function $\hat{f}$ that when queried for any element $x\in U$, returns a value $\hat{f}(x)$ such that}
    $$\forall x \in U, |f(x) - \hat{f}(x)| \leq \epsilon F_1,$$
    {where $F_1 = \sum_{x\in U} f(x)$.}
\end{definition}

\begin{definition} [Heavy Hitters problem]
{Given a parameter $\epsilon>0$, a {\em heavy hitter} is an item with frequency at least $\epsilon F_1$. Given a stream of operations (insertions or deletions) of items from a universe $U$, the {\em heavy hitters problem}  asks for an (explicitly defined) subset of $U$ that contains all heavy hitters.
}
\end{definition}

{All of our streaming algorithms use a {\em summary} 
 of the information seen in the stream so far, consisting of a collection of items and some counts for each of them. The size of the summary is much smaller than $|U|$ (equivalent to $n$), and we say that an item is {\em monitored} if it is present in the summary.}
 
\subsection{The SpaceSaving± Family of Algorithms}

\begin{algorithm}[]
\SetAlgoLined
    \KwIn{Insertion-only data stream $\sigma$ and SpaceSaving summary $S$}
    \KwOut{SpaceSaving summary $S$}
    
    \For{(item $e$, insertion) from $\sigma$}{
        \uIf{$e \in S$}{
            $count_{e}$ += 1 \;
        }
        \uElseIf{$S$ not full}{
            add $e$ to $S$ with $count_{e}$ = 1 \;
        }
        \uElse{
            $minItem$ = item $x \in S$ minimizing $count_{x}$\;
            $w$ = $count_{minItem}$\;
            evict $minItem$ from summary\;
            add $e$ to summary with $count_{e} = w+1$\;
        }
    }
    % {Result: the summary with one count for each item in summary.}
 \caption{SpaceSaving Update Algorithm (insertion-only, one count per item)}
 \label{SpaceSaving algorithm}
\end{algorithm}

The SpaceSaving$\pm$ family of algorithms are data summaries that build on top of the original SpaceSaving~\cite{metwally2005efficient} algorithm. In 2005, Metwally, Agrawal, and El~Abbadi~\cite{metwally2005efficient} proposed the popular \textbf{SpaceSaving} algorithm that provides highly accurate estimates on item frequency and heavy hitters among many competitors~\cite{cormode2008finding}. The SpaceSaving algorithm operates in the insertion-only model {(so that the stream is simply a sequence of items from $U$)} and uses the optimal space~\cite{berinde2010space}. The update algorithm, as shown in Algorithm~\ref{SpaceSaving algorithm}~\cite{metwally2005efficient}, uses $m$ counters to store a monitored item's identity and its estimated frequency. When a new $item$ arrives: if $item$ is monitored, then increment its counter; if $item$ is not monitored and the summary has unused counters, then start to monitor $item$, and set $count_{item}$ to 1; otherwise, i) find $minItem$, the item with minimum counter ($minCount$), ii) evicts this $minItem$ from the summary, and iii) monitor the new $item$ and set $count_{item}$ to $minCount+1$. {See Algorithm~\ref{SpaceSaving algorithm}.}

{After the stream has been read, the SpaceSaving Update algorithm has produced a summary. 
To solve the Frequency Estimation Problem, the algorithm  then needs to estimate the frequency $f(e)$ of an item $e$ being queried: the SpaceSaving Query algorithm returns $\hat f(e)=count_{e}$ if $e$ is monitored and $\hat f(e)=0$} otherwise, as shown in Algorithm~\ref{SpaceSaving query}. {(To solve the heavy hitters problem, the algorithm outputs all items in the summary whose count is greater than or equal to $\epsilon $ times the sum of the counts.)} 

When the number of counters {(size of the summary)} $m=\frac{1}{\epsilon}$, SpaceSaving solves both frequency estimation and heavy hitters problems in the insertion-only model, as shown in Lemma~\ref{spacesaving error bound}. Moreover, SpaceSaving satisfies a {more refined} residual error bound as shown in Lemma~\ref{spacesaving residual error bound} in which {\em "hot"} items do not contribute error and the error bound only depends on the {\em "cold"} and {\em "warm"} items~\cite{berinde2010space}.

The \textbf{$\epsilon$ error guarantee} is shown in Lemma~\ref{spacesaving error bound}. When $m=\frac{1}{\epsilon}$, SpaceSaving ensures $\forall x \in U, |f(x) - \hat{f}(x)| \leq \epsilon F_1$.
\begin{lemma} \label{spacesaving error bound}~\cite{metwally2005efficient} 
    After processing insertion-only data stream $\sigma$ {with the SpaceSaving Update algorithm (Algorithm~\ref{SpaceSaving algorithm}) with $m$ counters, the SpaceSaving Query algorithm (algorithm~\ref{SpaceSaving query}) provides an estimate $\hat f$ such that}  $\forall x \in U, |f(x) - \hat{f}(x)| \leq \frac{F_1}{m}$, where $F_1 = \sum_{i=1}^{n} f_i$.
\end{lemma}

The \textbf{residual error} guarantee is shown in Lemma~\ref{spacesaving residual error bound}. 
%When $m=O(\frac{k}{\epsilon})$, SpaceSaving ensures $\forall x \in U, |f(x) - \hat{f}(x)| \leq \frac{\epsilon}{k} F_{1}^{res(k)}$.
\begin{lemma} \label{spacesaving residual error bound} \cite{berinde2010space} 
    After processing insertion-only data stream $\sigma$, SpaceSaving with $O(\frac{k}{\epsilon})$ counters ensure $\forall x \in U, |f(x) - \hat{f}(x)| \leq \frac{\epsilon}{k} F_{1}^{res(k)}$ 
    where $F_{1}^{res(k)} = F_1 - \sum_{i=1}^{k} f_i$.
\end{lemma}

The \textbf{relative error} guarantee is shown in Lemma~\ref{spacesaving relative error bound}. 

\begin{definition}[$\gamma$-decreasing]
\label{gamma-decreasing}
Let $\gamma>1$, then a non-decreasing function with domain $\{1,2...n\}$ is $\gamma$-decreasing if for all t such that $1 \leq \gamma t \leq n$, $f_{\lceil \gamma t \rceil} \leq f_t/2$
\end{definition}

\begin{lemma} \label{spacesaving relative error bound} \cite{RC23} 
    After processing insertion-only data stream $\sigma$ and assume the exact frequencies follow the $\gamma$-decreasing function (Definition~\ref{gamma-decreasing}),
    SpaceSaving with $O(\frac{k}{\epsilon} \frac{2 -\gamma}{2(\gamma-1)})$ counters ensure $|f_i-\hat{f}_i| \leq \epsilon f_{i}, \forall i \leq k$.
\end{lemma}

\begin{algorithm}[]
\SetAlgoLined
    \KwIn{SpaceSaving summary $S$ and an item $e$}
    \KwOut{$e$'s estimated frequency}
        \uIf{$e \in S$}{
            return $count_{e}$\;
        }
        return 0\;
 \caption{Query Algorithm}
 \label{SpaceSaving query}
\end{algorithm}

\medskip

In 2021, Zhao et al.~\cite{zhao2021spacesaving} proposed SpaceSaving$\pm$ to extend the SpaceSaving algorithm from the insertion-only model to the bounded deletion model. {In that model, }they established a space lower-bound of { at least $\frac{\alpha}{\epsilon}$} counters required to solve the frequency estimation and heavy hitters problems, and \claire{proposed an algorithm,} SpaceSaving$\pm$ (see Algorithm~\ref{SpaceSaving+/- algorithm}), that used space matching the lower bound. The update algorithm is shown in Algorithm~\ref{SpaceSaving+/- algorithm} in which insertions are dealt with as in  Algorithm~\ref{SpaceSaving algorithm}, deletions of monitored items decrement the corresponding counters, and deletions of unmonitored items are ignored. The query algorithm is identical to Algorithm~\ref{SpaceSaving query}\footnote{There are a couple of variants of Algorithm~\ref{SpaceSaving query}. We adopt the most commonly used approach.}.

\begin{algorithm}[]
\SetAlgoLined
    \KwIn{Bounded deletion data stream $\sigma$, and SpaceSaving$\pm$ summary $S$}
    \KwOut{SpaceSaving$\pm$ Summary $S$}
    \For{ (item $e$, operation $op$) from $\sigma$}{
        \uIf{$op$ is an insertion}{
            follow Algorithm~\ref{SpaceSaving algorithm} to process $e$\;
        }
        \uElse{
            \uIf{$e \in S$}{
                $count_{e} -= 1$\;
            }
        }
    }
 \caption{SpaceSaving$\pm$ Update Algorithm (bounded deletion)}
 \label{SpaceSaving+/- algorithm}
\end{algorithm}

However, the SpaceSaving$\pm$ algorithm \claire{is only correct if} an implicit assumption holds~\cite{zhao2023errata}, \claire{namely,} that deletions can only happen after all the insertions, so that interleaving insertions and deletions are not allowed, as shown in Lemma~\ref{spacesaving+/- guarantee}. \claire{In this paper}, we present new algorithms \claire{in the general bounded deletions model}, that allow interleavings between insertions and deletions while \claire{still} using \claire{only} $O(\frac{\alpha}{\epsilon})$ space.

This Lemma correspond to Theorem~2 in~\cite{zhao2021spacesaving}.
Consider the bounded deletions model and assume, in addition, that the stream consists of an insertion phase with $I$ insertions followed by a deletion phase with $D$ deletions (no interleaving).
\begin{lemma} \label{spacesaving+/- guarantee}
\claire{After processing  data stream $\sigma$ with the SpaceSaving$\pm$ Update algorithm (Algorithm~\ref{SpaceSaving+/- algorithm})  with $m=\frac{\alpha}{\epsilon}$ counters, the SpaceSaving Query algorithm (algorithm~\ref{SpaceSaving query}) provides an estimate $\hat f$ such that}  $\forall x \in U, |f(x) - \hat{f}(x)| \leq \frac{F_1}{m}$, where $F_1 = \sum_{i=1}^{n} f_i$.
\end{lemma}
\section{New Algorithms}
\label{sec-alg}
In this section, we propose new algorithms supporting the bounded deletion model (with interleavings between insertions and deletions, while an item's frequency is never negative). Double SpaceSaving$\pm$ and Integrated SpaceSaving$\pm$ both use $O(\frac{\alpha}{\epsilon})$ space to solve the deterministic frequency estimation and heavy hitters problems in the bounded deletion model. We extended Double SpaceSaving$\pm$ to support unbiased estimations. Also, we establish the correctness proof for these algorithms.
While these algorithm share similar characteristics, they represent different trade-offs. In particular,  Double SpaceSaving$\pm$ uses two disjoint summaries, one to track insertions and the other to track deletions, while Intergrated SpaceSaving$\pm$ uses a single summary to track both deletions and insertions.

%The main correction applied to SpaceSaving$\pm$ is to ensure the minimum count inherit during eviction process is monotonically increasing.

\subsection{Double SpaceSaving$\pm$}

\begin{algorithm}[]
\SetAlgoLined
    \KwIn{Bounded deletion data stream $\sigma$, and Double SpaceSaving$\pm$ summary ${(S_{insert},S_{delete})}$}
    \KwOut{Double SpaceSaving$\pm$ Summary $S{=(S_{insert},S_{delete})}$}
   
    \For{(item $e$, operation $op$) from $\sigma$}{
        \uIf{$op$ is an insertion}{
            {Apply Algorithm~\ref{SpaceSaving algorithm} to insert $e$ in $S_{insert}$}\;
        }
        \uElse{
            {Apply Algorithm~\ref{SpaceSaving algorithm} to insert $e$ in $S_{delete}$}\;
        }
    } 
  \caption{Double SpaceSaving$\pm$ Update Algorithm}
  \label{double-spacesaving-update}
\end{algorithm}

\begin{algorithm}[]
\SetAlgoLined
    \KwIn{Double SpaceSaving$\pm$ summary $S$ and an item $e$}
    \KwOut{$e$'s estimated frequency.}
    inserts = {Query $S_{insert}$ for $e$ with Algorithm~\ref{SpaceSaving query}}\;
    deletes = {Query $S_{delete}$ for $e$ with Algorithm~\ref{SpaceSaving query}}\;
        {return max(inserts - deletes, 0)}\;
 \caption{Double SpaceSaving$\pm$ Query Algorithm}
 \label{double-spaceSaving-query}
\end{algorithm}

In Double SpaceSaving$\pm$, we employ two SpaceSaving summaries to track insertions ($S_{insert}$) and deletions ($S_{delete}$) separately, using $m_{I}$ and $m_{D}$ counters respectively (the algorithms is parameterized by the choice of $m_I$ and $m_D$). As shown in Algorithm~\ref{double-spacesaving-update}, Double SpaceSaving$\pm$ completely partitions the update operations to two independent summaries.
To query for the frequency of an item $e$, both summaries need to be queried to estimate the frequency of item $e$ as shown in Algorithm~\ref{double-spaceSaving-query}.

\begin{theorem} \label{dss theorem: frequency estimation}
After processing a bounded deletion data stream $\sigma$, Double SpaceSaving$\pm$ following Update Algorithm~\ref{double-spacesaving-update} with $m_I=O(\frac{\alpha}{\epsilon})$ and $m_D=O(\frac{\alpha-1}{\epsilon})$  and Query Algorithm~\ref{double-spaceSaving-query} provides an estimate $\hat{f}$ such that $\forall x \in U, |f(x) - \hat{f}(x)| < \epsilon F_1$.
\end{theorem}

\begin{proof}
    For every items in the universe their frequency estimation error is upper bounded by $\frac{I}{m_{I}} + \frac{D}{m_{D}}$ due to the property of SpaceSaving as shown in Lemma~\ref{spacesaving error bound} and triangle inequality. 
    In addition, from the definition of bounded deletion, we can derive the relationship between $I$, $D$, and $F_1$. Namely, $I \leq \alpha F_1$ and $D \leq (\alpha-1) F_1$. 
    Set $m_{I} = \frac{2\alpha}{\epsilon}$ and $m_{D}=\frac{2(\alpha-1)}{\epsilon}$, 
    then $\frac{I}{m_{I}} + \frac{D}{m_{D}} \leq \epsilon F_1$. 
    Hence, Double SpaceSaving$\pm$ solve the deterministic frequency estimation problem in the bounded deletion model with $O(\frac{\alpha}{\epsilon})$ space.
\end{proof}

\begin{theorem}
\label{dss theorem: Frequent}
Double SpaceSaving$\pm$ solves the heavy hitters problem in the bounded deletion model using $O(\frac{\alpha}{\epsilon})$ space by reporting all items monitored in $S_1$.
\end{theorem}

\begin{proof} Assume a heavy hitter $x$ is not monitored in the summary. By definition of heavy hitters, $f(x) \geq \epsilon F_1$. Since $x$ is not monitored, its estimated frequency 0. As a result, the estimation error for item $x$ is larger than $\epsilon F_1$ which contradicts Theorem~\ref{dss theorem: frequency estimation}.
  
Hence, by contradiction Double SpaceSaving$\pm$ solves the heavy hitters problem.
\end{proof}

Being more careful and using variants of Algorithm~\ref{SpaceSaving query} would lead to improvements but would not change the theoretical bound of Theorem~\ref{dss theorem: Frequent}.

\subsubsection{Unbiased Double SpaceSaving$\pm$}
\label{sec:udss}
Unbiased estimation is a desirable property for robust estimation~\cite{charikar2002finding, duffield2007priority}. Research has shown that SpaceSaving can be extend to provide unbiased estimation (\claire{i.e.}, $E[\hat{f}(x)] = f(x)$)~\cite{ting2018data}. 
\claire{The extension is obtained from Algorithm 1 by the following simple modification: after line 8 of the code, with probability $1/(w+1)$ lines 9 and 10 are executed; with the complementary probability $1-1/(w+1)$, instead of lines 9 and 10 the algorithm simply increments $count_{minItem}$.} This indicates that we can replace the two independent deterministic SpaceSaving summaries with two independent unbiased SpaceSaving summaries in Algorithm~\ref{double-spaceSaving-query} to support unbiased frequency estimation in the bounded deletion model. 

\begin{theorem}
\label{dss theorem: unbiased}
After processing a bounded deletion data stream $\sigma$, Unbiased Double SpaceSaving$\pm$ of space $O\frac{\alpha}{\epsilon}$, with two independent unbiased SpaceSaving summaries in Algorithm~\ref{double-spaceSaving-query}, provides unbiased frequency estimation and the estimation variance is upper bounded by $\epsilon^2 F_1^2$.
\end{theorem}

\begin{proof} 
The variance of an Unbiased SpaceSaving's estimation is upper bounded by $O(\frac{F_1}{m}^2)$~\cite{ting2018data}. We know that $Var[X-Y]$ = $Var[X]$ + $Var[Y]$ - $2(Cov[X,Y]$, where $X$ is the estimation of inserts and $Y$ is the estimation of deletes. Since two summaries are independent, the covariance is 0; Hence, the variance upper bounded by the sum of the $O(\frac{I}{m_{I}}^2 + \frac{D}{m_{D}}^2)$. In Theorem~\ref{dss theorem: frequency estimation}, we proved that $\frac{I}{m_{I}} + \frac{D}{m_{D}} \leq \epsilon F_1$. As a result, the variance of the unbiased estimation for any item's frequency is upper bounded by $\epsilon^2 F_1^2$.

\end{proof}

\claire{\subsection{Integrated SpaceSaving$\pm$}}
\begin{algorithm}[]
\SetAlgoLined
    \KwIn{Bounded deletion data stream $\sigma$, and IntergratedSpaceSaving$\pm$ summary $S$}
    \KwOut{IntergratedSpaceSaving$\pm$ summary $S$}
    \For{(item $e$, operation $op$) from $\sigma$}{
        \uIf{$e \in S$}{
            \uIf{$op$ is an insertion}{
                $insert_{e}$ += 1\;
            }
            \uElse{
                $delete_{e}$ += 1\;
            }
        }
        \uElse{
            \uIf{$S$ not full}{
                % //$e$ must be an inserted since delete an uninserted item is not allowed\;
                $insert_{e}$ = 1\;
                $delete_{e}$ = 0\;
            }
            \uElseIf{$op$ is an insertion}{
                $minItem$ = item $x \in S$ minimizing $insert_{x}$\;
                $w$ = $insert_{minItem}$\;
                evict $minItem$ from summary\;
                add $e$ to summary and set $insert_{e} = w+1$ and $delete_{e}=0$\;
                
            }
        }
    }
  \caption{Integrated SpaceSaving$\pm$ Update Algorithm}
  \label{IntegratedSpaceSaving Update}
\end{algorithm}

Unlike Double SpaceSaving$\pm$ which uses two independent summaries, Integrated SpaceSaving$\pm$ completely integrates deletions and insertions in one data structure as shown in Algorithm~\ref{IntegratedSpaceSaving Update}. The main difference between Integrated SpaceSaving$\pm$ and SpaceSaving$\pm$ is that Integrated SpaceSaving$\pm$ tracks the insert and delete count separately. This ensures that the minimum insert count monotonically increases and never decrease; however, the original SpaceSaving$\pm$ uses one single count to track inserts and deletes together. When insertions and deletions are interleaved, then the minimum count in SpaceSaving$\pm$ may decreases. This can lead to severe underestimation of a frequent item~\cite{zhao2023errata}.

\begin{algorithm}[]
\SetAlgoLined
    \KwIn{Integrated SpaceSaving$\pm$ summary $S$ and an item $e$}
    \KwOut{$e$'s estimated frequency.}
    \uIf{$e \in S$}{
            return $insert_{e} - delete_{e}$\;
    }
    return 0\;
 \caption{Integrated SpaceSaving$\pm$ Query Algorithm}
 \label{IntegratedSpaceSaving Query}
\end{algorithm}

More specifically, the Integrated SpaceSaving$\pm$ data structure maintains a record estimating the number of insertions and deletions for each monitored item. When a new item arrives, if it is already being monitored, then depending on the type of operation, either the insertion or the deletion count of that item is incremented.  If the summary is not full, then no item could have every been evicted from the summary, and since no delete operation arrives before the corresponding item is inserted, then any new item in the stream must be an insertion. Hence, when a new item arrives, and if the summary is not full, the new item is added to the summary and initialized with one insertions and 0 deletions. If a delete arrives and the summary is full, the delete is simply ignored. The interesting case is when an insert arrives and the summary is full.  In this case, Integrated SpaceSaving$\pm$ mimics the behaviour of the original SpaceSaving$\pm$ by evicting the element with the least number of inserts and replacing it with the new element. The insertion count for the new element is overestimated by initializing it to the insertion count for the evicted element. However, the deletion count is underestimated by initializing the deletion count to zero.

To estimate an item's frequency, if the item is monitored, it subtracts the item's delete count from the item's insert count, otherwise the estimated frequency is 0, as shown in Algorithm~\ref{IntegratedSpaceSaving Query}.

We now establish the correctness of Integrated SpaceSaving$\pm$. First, we establish three lemmas about Integrated SpaceSaving$\pm$ to help us prove the correctness of the algorithm.

{
\begin{lemma} \label{lemma: sum of insert counts equal to I}
After processing a bounded deletion data stream $\sigma$, Integrated SpaceSaving$\pm$  following update Algorithm~\ref{IntegratedSpaceSaving Update} ensures that the sum of all insert counts equals $I$.
\end{lemma} 

\begin{proof}
    This holds by induction over time, as verified by inspection of the algorithm.
\end{proof}
}

\begin{lemma} \label{lemma: upperbound on minCount for spacesaving}
After processing a bounded deletion data stream $\sigma$, Integrated SpaceSaving$\pm$ with $m$ counters following update Algorithm~\ref{IntegratedSpaceSaving Update} ensures that the minimum insert count, $insert_{minItem}$, is less than or equal to $\frac{I}{m}$.
\end{lemma}

\begin{proof}
Since deletions never increment any insert counts, insert counts is only affected by insertions. We observe that the sum of all insert counts in summary is exactly equal to $I$, since the sum of all insert counts increment by 1 after processing one insertion. As a result, $insert_{minItem}$ is maximized when all insert counts are the same. Hence, $insert_{minItem}$ is less than or equal to $\frac{I}{m}$.
\end{proof}

\begin{lemma} \label{lemma: overestimation for IntegratedSpaceSaving}
All monitored items in Integrated SpaceSaving$\pm$ are never underestimated following query Algorithm~\ref{IntegratedSpaceSaving Query}.
\end{lemma}
\begin{proof}
{The $insert_e$ counts of the summary are dealt with exactly like in the SpaceSaving algorithm, for the substream of $\sigma$ consisting of the insertion operations only. By the analysis of SpaceSaving, it is known that for an item $e$ in the summary, $insert_e$ is an overestimate of the true number of insertions of $e$ since the beginning of the stream. 

The $delete_e$ count of element $e$ in the summary is equal to the number of deletions of item $e$ since the last time that it was inserted into the summary, so it is an underestimate of the  number of deletions of $e$ since the beginning of the stream.

Thus, for any $e$ in summary, $insert_e-delete_e$ is an overestimate.}
\iffalse
First, we observe that deletions cannot lead to underestimation since the deletion either brings down both the true frequency and the estimated frequency by one (for monitored items), or is completely ignored (for unmonitored items). As a result, we only need to consider the effect of processing insertions. Without loss of generality, let's consider an arbitrary item $x$ in the summary and consider how $insert_{x}$ evolves. 

For the first time when $x$ is inserted. The $insert_{x, 1}$ is larger or equal to 1, and $x$'s true frequency is 1. Assume, at the $n^{th}$ time when $x$ is inserted The $insert_{x, n}$ is larger or equal to $n$. Before the $(n+1)^{th}$ insertion, there are two cases 1) $x$ is monitored and 2) $x$ is not monitored. 

If $x$ is monitored before the $(n+1)^{th}$ insertion, then its corresponding counter will add 1 after $(n+1)^{th}$ insertion. Since $insert_{x, n}$ is larger or equal to $n$, $insert_{x, n+1}$ must be larger or equal to $n+1$.  

If $x$ is not monitored, $x$ must be evicted after the $n^{th}$ insertion of $x$. In particular, at some time in between $n^{th}$ to $(n+1)^{th}$ insertion of $x$, $insert_{x, n}$ becomes the minimum insert count. After that eviction, no insert count in the summary can be less than $insert_{x, n}$. As a result, no insert count in the summary at $(n+1)^{th}$ insertion of $x$ can be less than $insert_{x,n}$. Hence, after process $(n+1)^{th}$ insertion of $x$, $insert_{x, n+1}$ must be larger or equal to $n+1$.
\fi 
\end{proof}

\begin{lemma}
    \label{lemma iss-above}
    In Integrated SpaceSaving$\pm$, \michel{any} item inserted more than $insert_{minItem}$ must be monitored in the summary; {moreover, for any item $e$ in the summary, its count $insert_e$ exceeds its true number of insertions by at most $insert_{minItem}$.}
\end{lemma}

\begin{proof}
{The $insert_e$ counts of the summary are dealt with exactly like in the SpaceSaving algorithm, for the substream of $\sigma$ consisting of the insertion operations only, and the properties are well-known to hold for SpaceSaving.}
\iffalse
Proof by Contradiction. Assume there exists an item $x$ which is being inserted more than $insert_{minItem}$ times and $x$ is not monitored by the summary. Since $x$ is not monitored, then $x$ must be evicted before this time. The eviction is based on the minimum insert count. Moreover, insert count is monotonically increasing (never decreases). As a result, if $x$ is evicted its insert count must be less than or equal to the $insert_{minItem}$. However, this leads to a contradiction.
\fi
\end{proof}

\begin{lemma}
    \label{lemma, IntegratedSpaceSaving Error Upper Bound}
    In Integrated SpaceSaving$\pm$ following update Algorithm~\ref{IntegratedSpaceSaving Update} and query Algorithm~\ref{IntegratedSpaceSaving Query}, the estimation error for any item is upper bounded by $insert_{minItem}$. \michel{ We can then write:} $\forall x \in U, |f(x) - \hat{f}(x)| \leq insert_{minItem}$.
\end{lemma}

\begin{proof}
{First consider the case when $e$ is unmonitored, so the estimate is $0$. Its true frequency is at most its number of insertions, and by Lemma~\ref{lemma iss-above} that is bounded by $insert_{minItem}$, hence the Lemma holds for $e$.

Now, consider the case when $e$ is in the summary, and consider the last time $t$ at which $e$ was inserted in the Summary. Let $I_{\geq t}(e)$ and $D_{\geq t}(e)$ denote the number of insertions and of deletions of $e$ since time $t$, and let $min$ denote the value of $insert_{minItem}$ at time $t$. Since the algorithm updates $insert_e$ and $delete_e$ correctly while $e$ is in the summary, the query algorithm outputs the estimate $\hat f(e)=insert_e-delete_e=min+I_{\geq t}(e)-D_{\geq t}(e)$. On the other hand, the frequency $f(e)$ of $e$ equals the net number of occurences of $e$ prior to time $t$, plus $I_{\geq t}(e)-D_{\geq t}(e)$. By Lemma~\ref{lemma iss-above} the net number of occurences of $e$ prior to time $t$ is at most $min$, so $\hat f(e)-min\leq f(e)\leq \hat f(e)$, implying the Lemma. }
\iffalse
Insert of monitored items, deletion of both monitored and unmonitored items do not lead to any additional error. The reason deleting an unmonitored item does not lead to any additional error is that the estimated frequency for unmonitored item is 0 and deletion decreases its true frequency; hence, the estimated error for this particular item decreases. As a result, the error in the summary are really from eviction. The error for the evicted item is at most $insert_{minItem}$. This is because the true frequency of the evicted item is less than or equal to $insert_{minItem}$ (as shown in Lemma~\ref{lemma: overestimation for IntegratedSpaceSaving}) and after eviction its estimated frequency is now 0. Moreover, the estimation error for the newly inserted item is also $insert_{minItem}$. Since the new item is not monitored, then the new item's true frequency is clearly no larger than $insert_{minItem}$ based on Lemma~\ref{lemma iss-above} and no less than 1.  After processing the insertions, the new item's estimated frequency is now $insert_{minItem}+1$. As a result, the estimation error for the new item is at most $insert_{minItem}$.
\fi 
\end{proof}

\begin{theorem} \label{theorem: frequency estimation}
After processing a bounded deletion data stream $\sigma$, IntegratedSpaceSaving$^{\pm}$ with $O(\frac{\alpha}{\epsilon})$ space following update Algorithm~\ref{IntegratedSpaceSaving Update} and query Algorithm~\ref{IntegratedSpaceSaving Query} provide an estimate $\hat{f}$ such that $\forall i, |f(i) - \hat{f}(i)| < \epsilon F_1$.
\end{theorem}

\begin{proof}
By Lemma~\ref{lemma: upperbound on minCount for spacesaving} and Lemma~\ref{lemma, IntegratedSpaceSaving Error Upper Bound}, we known that $\forall x \in U, |f(x) - \hat{f}(x)| \ \leq insert_{minItem} \leq \frac{I}{m}$. We also know that $I \leq \alpha F_1$. Let's set $m = \frac{\alpha}{\epsilon}$, then we have $\frac{I}{m} \leq \epsilon F_1$. As a result, IntegratedSpaceSaving$^{\pm}$ with $O(\frac{\alpha}{\epsilon})$ space ensures $\forall i, |f(i) - \hat{f}(i)| < \epsilon F_1$.
\end{proof}

Integrated SpaceSaving$\pm$ also solves the heavy hitters problem. Since Integrated SpaceSaving$\pm$ never underestimates monitored items (Lemma~\ref{lemma: overestimation for IntegratedSpaceSaving}), then if we report all the items with frequency estimations greater than or equal to $\epsilon F_1$, then all heavy hitters will be identified as shown in Theorem~\ref{theorem: Frequent}.
 
\begin{theorem}
\label{theorem: Frequent}
IntegratedSpaceSaving$^{\pm}$ solves the heavy hitters problem in the bounded deletion model using $O(\frac{\alpha}{\epsilon})$ space.
\end{theorem}

\begin{proof} Assume a heavy hitter $x$ is not contained if the set of all items with estimated frequency larger than or equal to $\epsilon F_1$. Recall, by definition of heavy hitters, $f(x) \geq \epsilon F_1$. Since $x$ is not contained, $x$'s frequency estimation, $\hat{f}(x)$, must be less than $\epsilon F_1$. However, by Lemma~\ref{lemma: overestimation for IntegratedSpaceSaving}, Integrated SpaceSaving$\pm$ never underestimates the frequency of monitored items.
  
Hence, by contradiction Integrated SpaceSaving$\pm$ solves the deterministic heavy hitters problem.
\end{proof}

\subsection{Comparisons between Double SpaceSaving$\pm$ and Integrated SpaceSaving$\pm$}
We assume each field uses the same number of bits (32-bit in the implementation).
Following the proofs of Theorem~\ref{theorem: frequency estimation} and that each entry in the summary has three fields $(i, insert_i, delete_i)$, IntegratedSpaceSaving$\pm$ uses $3 \frac{\alpha}{\epsilon}$. On the other hand, following the proofs of Theorem~\ref{dss theorem: frequency estimation} and that each entry in the summary uses two fields, Double SpaceSaving$\pm$ uses $2\frac{2\alpha}{\epsilon} + 2\frac{2(\alpha-1)}{\epsilon} = \frac{8\alpha-2}{\epsilon}$. Integrated SpaceSaving$\pm$ requires less memory footprint. Double SpaceSaving$\pm$ also has its own advantage and it can be extended supporting unbiased estimations (Unbiased Double SpaceSaving$\pm$).

% There are examples where Double SpaceSaving$\pm$ may be better at identifying warm items~\footnote{Warm items refer to items, though with count less than $\epsilon F_1$, that appear more often than other items.} more accurately.

% We also remark that Algorithms~
% \ref{double-spacesaving-update} and~\ref{double-spaceSaving-query} can easily be extended to the setting in which the number of deletions of an item is allowed to be greater than the number of insertions of an item: to handle that case, it suffices to remove the clipping in the last line of Algorithm~\ref{double-spaceSaving-query}.
\section{Tighter Analysis on Error Bounds}
In this section, we present tighter analysis of the proposed algorithms. In Section~\ref{sec-alg}, we showcased that the proposed algorithms achieve the desired error bound $\epsilon F_1$, in which the error bound depends on all items in data stream $\sigma$. In this section, we first prove that both algorithms guarantee the residual error bound and then prove that they both also ensure the relative error bound under relatively mild assumptions. We take inspirations from~\cite{berinde2010space, metwally2005efficient, zhao2021spacesaving, RC23}.
%(and a paper under submission).

\subsection{Residual Error Bound}
We define $F_{1, \alpha}^{res(k)} = F_1 - \frac{1}{\alpha}\sum_{i=1}^{k} f_{i}$. Observe that $F_{1, \alpha}^{res(k)}$ is equivalent to the residual error bound shown in Lemma~\ref{spacesaving residual error bound}, $F_{1}^{res(k)}$, for the insertion-only model, as $\alpha$ is set to 1. The residual error bound ($F_{1, \alpha}^{res(k)}$) is much tighter than the original error bound ($F_1$), since the residual error depends less on the heavy hitters.

First, we derive the residual error bound for Double SpaceSaving$\pm$. 
\begin{theorem}
    \label{dss+- residual error}
    Double SpaceSaving$\pm$ with $m_{I}+m_{D} = O(\alpha k / \epsilon)$ space achieves the residual error bound guarantee in which $\forall x \in U, |f(x) - \hat{f}(x)| \leq \frac{\epsilon}{k} F_{1, \alpha}^{res(k)}$, where $k < min(m_{I}, m_{D})$.
\end{theorem}

\begin{proof}
    First, let's define the maximum error $\delta$ as $\delta = max(\forall x \in U, |f(x) - \hat{f}(x)|)$ for Double SpaceSaving$\pm$. 
    We know that $\delta \leq \frac{I-\sum_{i=1}^{k}I_{i}}{m_{I}-k} + \frac{D-\sum_{i=1}^{k}D_{i}}{m_{D}-k}$ by Lemma~\ref{spacesaving residual error bound} and triangle inequality. By definition of bounded deletion model, $I \leq \alpha F_1$ and $D \leq (\alpha-1) F_1$. 
    Hence, we have $\delta \leq  \alpha \frac{F_1-1/\alpha \sum_{i=1}^{k}I_{i}}{m_{I}-k} + (\alpha-1)\frac{F_1-1/(\alpha-1) \sum_{i=1}^{k}D_{i}}{m_{D}-k} $. We know that $\sum_{i=1}^{k}I_{i} \geq \sum_{i=1}^{k}f_{i}$ (this is because sum of k largest insert count has to be larger than or equal to the sum of k largest frequencies) and $\sum_{i=1}^{k}D_{i} \geq 0 $.
    %and $\sum_{h=1}^{k}D_{h} \geq \sum_{i=1}^{k}f_{i} - \sum_{i=1}^{k}f_{i}$ = 0. 
    As a result, $\delta <  \alpha \frac{F_1-1/\alpha \sum_{i=1}^{k}f_{i}}{m_{I}-k} + (\alpha-1) \frac{F_1 }{m_{D}-k}$. Let $m_{I} = k (\frac{2\alpha}{\epsilon}+1)$ and $m_{D} = k (\frac{2(\alpha-1)}{\epsilon}+1)$. Then, $\delta < \frac{\epsilon}{k} (F_1 - 1/(2\alpha)\sum_{i=1}^{k}f_{i}) \approx \frac{\epsilon}{k} F_{1, \alpha}^{res(k)}$.
\end{proof}

We present the residual error bound for Integrated SpaceSaving$\pm$. Before presenting the proof, we first construct a helpful Lemma.

\begin{lemma}
    \label{lemma, iss-insert-mincount}
    The minimum insert count ($insert_{minItem}$), in Integrated SpaceSaving$\pm$ with m counters, is less than or equal to $\alpha \frac{F_1 - 1/\alpha \sum_{i=1}^{k}f_i}{m-k}$ where $k<m$.
\end{lemma}

\begin{proof}
    The sum of all insert count equals to $I$, as the sum of all insert count always increase by 1 after processing an insertion. If we zoom into the k largest insert count in the summary, the sum of their insert counts is no less than $\sum_{i=1}^{k}f_i$, since no monitored items are underestimated as shown in Lemma~\ref{lemma: overestimation for IntegratedSpaceSaving}. 
    As a result, $insert_{minItem} \leq \frac{I - \sum_{i=1}^{k}f_i}{m-k}$. We know that by definition of bounded deletion, $I\leq \alpha F_1$. Hence, $insert_{minItem} \leq \alpha \frac{F_1 - 1/\alpha \sum_{i=1}^{k}f_i}{m-k}$.
\end{proof}

\begin{theorem}
    \label{ss+- residual error}
    Integrated SpaceSaving$\pm$ with $O(\alpha k / \epsilon)$ space achieves the residual error bound guarantee in which $\forall x \in U, |f(x) - \hat{f}(x)| \leq \frac{\epsilon}{k} F_{1, \alpha}^{res(k)}$, where $k<m$.
\end{theorem}

\begin{proof}
    From Lemma~\ref{lemma, IntegratedSpaceSaving Error Upper Bound} and Lemma~\ref{lemma, iss-insert-mincount}, we know the estimation error for all items is upper bounded by $insert_{minItem}$ and $insert_{minItem} \leq \alpha \frac{F_1 - 1/\alpha \sum_{i=1}^{k}f_i}{ m-k}$. Let $m=k(\frac{\alpha}{\epsilon}+1)$. We achieve the desired residual error bound: $\forall x \in U, |f(x) - \hat{f}(x)| \leq \frac{\epsilon}{k} F_{1, \alpha}^{res(k)}$.
\end{proof}

\subsection{Relative Error Bound}
The \textbf{$\epsilon$ relative error guarantee} is defined such that $\forall x \in U, |f(x)-\hat{f}(x)| \leq \epsilon f(x)$. The relative error has been shown to be much harder to achieve, but with practical importance~\cite{gupta2003counting, cormode2021relative}. In a recent work, Mathieu and de Rougemont~\cite{RC23} showcased that the original SpaceSaving algorithm achieves relative error guarantees under mild assumptions. As a result, we demonstrate that the proposed algorithms in the bounded deletion model also achieve relative error guarantees under similar assumptions.

\begin{lemma}
    Let f be $\gamma$ decreasing with $1<\gamma<2$. 
    %Then $\sum_{j=i}^{n} f_j \leq i f_i \frac{2log_2 \gamma}{1-log_2 \gamma}$.
    Then $\sum_{j=i}^{n} f_{j} \leq (i-1)f_{i-1}\frac{2(\gamma-1)}{\gamma-1} $ 
    and hence $F_1 = \sum_{j=1}^{n} f_j \leq f_1 (1+\frac{2(\gamma-1)}{2-\gamma}) = f_1\frac{\gamma}{2-\gamma}$.
\end{lemma}
\begin{proof}
    We omit the proof. (This is proven in Lemma 5 at~\cite{RC23})
\end{proof}

\begin{definition}[Zipfian Distribution~\cite{zipf2016human}]
Let a function f with domain $\{1,2...n\}$ follow the Zipf distribution with parameter $\beta$, then $f_i = F_1 \frac{1}{i^{\beta} \xi(\beta)}$ where $F_1 = \sum_{i=1}^{n} f_i$ and $\xi(\beta) = \sum_{i=1}^{n} i^{-\beta}$.
\end{definition}
By the definition of Zipf distribution, we know that $\xi(\beta)$ converges to a small constant when $\beta > 1$ and $f_1 = F_1/\xi(\beta)$.

\begin{theorem}
    \label{dss+- relative error}
    Assuming insertions and deletions in the $\alpha$-bounded deletion model follow the $\gamma$-decreasing property ($1<\gamma<2$) and the exact frequencies follow the Zipf distribution with parameter $\beta$, then Double SpaceSaving$\pm$ with $ O(\frac{\alpha k^{\beta+1}}{\epsilon 2^{log_{\gamma}k}})$ space achieves the relative error bound guarantee such that 
    $|f_i - \hat{f}_i| \leq \epsilon f_i$, for $i\leq k$.
\end{theorem}

\begin{proof}
    Recall, the maximum estimation error $\delta$ is upper bounded by $\frac{I - \sum_{i=1}^{k}I_i}{m_{I}-k} + \frac{D - \sum_{i=1}^{k}D_i}{m_{D}-k}$. Since insertions and deletions follow $\gamma$-decreasing property, then $\sum_{i=k+1}^{n}I_i \leq k I_k \frac{2(\gamma-1)}{2-\gamma}$ and $\sum_{i=k+1}^{n}D_i \leq k D_k \frac{2(\gamma-1)}{2-\gamma}$. Hence, $\delta \leq 
    \frac{k I_k \frac{2(\gamma-1)}{2-\gamma}}{m_{I}-k} +  \frac{k D_k \frac{2(\gamma-1)}{2-\gamma}}{m_{D}-k}$.

% D < I(1-1/alpha)
% D < I - I/alpha
% I/alpha < I-D
% I < \alpha(I-D) = \alpha F_1
    For insertions, we know i) $I_1 \leq I \leq \alpha F_1 = \alpha \xi (\beta) f_1$, ii) $f_i = \frac{f_1}{i^{\beta}}$, and iii) $I_{i} \leq I_1 / 2^{log_{\gamma}(i)}$; For deletions, we know i) $D_1 \leq D \leq (\alpha-1) F_1 = (\alpha-1) \xi (\beta) f_1$, ii) $f_i = \frac{f_1}{i^{\beta}}$, and iii) $D_{i} \leq D_1 / 2^{log_{\gamma}(i)}$. 
    
    Hence, $I_k \leq \frac{I_1}{2^{log_{\gamma}(k)}} \leq \alpha \xi(\beta) f_1 / 2^{log_{\gamma}(k)} = \alpha \xi(\beta) f_k k^{\beta}  / 2^{log_{\gamma}(k)}$, and 
    $D_k \leq D_1 / 2^{log_{\gamma}(k)} \leq (\alpha-1) \xi(\beta) f_k k^{\beta}  / 2^{log_{\gamma}(k)}$.
    Let $m_{I} = m_{D}= k + \frac{2(\gamma-1)}{2-\gamma}\frac{k^{\beta+1}}{ 2^{log_{\gamma}k}} \frac{(2\alpha-1)}{\epsilon}$, then $\delta \leq \epsilon f_k$ and hence for all $i < k$, $|f_i - \hat{f}_i| \leq \epsilon f_i$.

    % For insertions, we know i) $I_1 \leq I \leq \alpha F_1$, ii) $f_i = \frac{f_1}{i^{\beta}}$, and iii) $I_{i} \leq I_1 / 2^{log_{\gamma}(i)}$; For deletions, we know i) $D_1 \leq D \leq (\alpha-1) F_1$, ii) $f_i = \frac{f_1}{i^{\beta}}$, and iii) $D_{i} \leq D_1 / 2^{log_{\gamma}(i)}$. 
    % Hence,  $I_k \leq I_1 / 2^{log_{\gamma}(k)} \leq \alpha F_1 / 2^{log_{\gamma}(k)}$, and 
    % $D_k \leq D_1 / 2^{log_{\gamma}(k)} \leq (\alpha-1) F_1 / 2^{log_{\gamma}(k)}$
    % Let $m_{I} = m_{D}= k + \frac{2(\gamma-1)}{2-\gamma}\frac{k^{\beta+1}}{ 2^{log_{\gamma}k}} \frac{(2\alpha-1)}{\epsilon}$, then $\delta \leq \epsilon f_k$ and hence for all $i < k$, $|f_i - \hat{f}_i| \leq \epsilon f_i$.
\end{proof}

\begin{theorem}
    \label{ss+- relative error}
    Assuming insertions in the $\alpha$-bounded deletion model follows the $\gamma$-decreasing property ($1<\gamma<2$) and the exact frequencies follow the Zipf distribution with parameter $\beta$, then Integrated SpaceSaving$\pm$ with $O(\frac{\alpha k^{\beta+1}}{\epsilon 2^{log_{\gamma}k}})$ space achieves the relative error bound guarantee such that 
    $|f_i - \hat{f}_i| \leq \epsilon f_i$, for $i\leq k$.
\end{theorem}

\begin{proof}
    We know that estimation error at most $insert_{minItem}$ and $insert_{minItem} \leq \frac{I - \sum_{i=1}^{k}I_i}{m-k}$. Since insertions follow $\gamma$-decreasing property, then $\sum_{i=k+1}^{n}I_i \leq k I_k \frac{2(\gamma-1)}{2-\gamma}$. Hence, $insert_{minItem} \leq 
    \frac{\sum_{i=k+1}^{n}I_i}{m-k} 
    \leq \frac{k I_k \frac{2(\gamma-1)}{2-\gamma}}{m-k}$. 
    Since we know i) $I_1 \leq I \leq \alpha F_1 = \alpha \xi (\beta) f_1$, 
    ii) $f_i = \frac{f_1}{i^{\beta}}$, and iii) $I_{g} \leq I_1 / 2^{log_{\gamma}(g)}$. 
    Hence,  $I_k \leq I_1 / 2^{log_{\gamma}(k)} \leq \alpha \xi(\beta) f_1 / 2^{log_{\gamma}(k)} = \alpha \xi(\beta) f_k k^{\beta}  / 2^{log_{\gamma}(k)}$.
    Let $m= k + \frac{2(\gamma-1)}{2-\gamma}\frac{k^{\beta+1}}{ 2^{log_{\gamma}k}} \frac{\alpha}{\epsilon}$, then $insert_{minItem} \leq \epsilon f_k$ and hence for all $i < k$, $|f_i - \hat{f}_i| \leq \epsilon f_i$.
\end{proof}

\section{Mergeability}
Mergeability is a desirable property in distributed settings in which two summaries on two data sets are given, after merging between these two summaries, the merged summary has error and size guarantees equivalent to a single summary which processed the union of two data sets. 
\begin{definition}[Mergeable~\cite{agarwal2013mergeable}]
    A summary $S$ is mergeable if there exists a merge algorithm $A$ that produces a summary $S(\sigma_1 \cup \sigma_2, \epsilon)$ from two input summaries $S(\sigma_1, \epsilon)$ and $S(\sigma_2, \epsilon)$. Note, the size of all three summaries should be the same in achieving $\epsilon$ error guarantee.
\end{definition}
%\claire{Moreover, if, subsequent to the merge, there are further updates to the merged streams, the Update algorithm must be able to proceed with the merged summary. Thus, the properties of the summary that were used to establish the Lemmas of Section 3.1 and of Section 3.2 \claire{(XXX Here, should put in the proper lemma numberings...)} must still hold for the merged summary.}

Both the original SpaceSaving and the Unbiased SpaceSaving have been shown to be mergable~\cite{agarwal2013mergeable, ting2018data}.
Since Double SpaceSaving$\pm$ relies on two independent summaries, Double SpaceSaving$\pm$ is also mergeable. 

For IntergratedSpaceSaving$\pm$, to define a Merge algorithm, we take inspiration from~\cite{agarwal2013mergeable}. The merge algorithm is shown in Algorithm~\ref{IntegratedSpaceSaving Merge}. Given two IntergratedSpaceSaving$\pm$ summaries of $m$ counters, $S_1$ and $S_2$,  we first union the two summaries (i,e, if both summaries monitor item $x$, then $x$'s insert and delete counts in the merge summary are the sum of $x$'s insert and delete counts from $S_1$ and $S_2$). After the union, the merged summary contain at most $2m$ counters. The merged summary consists of the $m$ largest items among these 2$m$ counters based on the insert counts. 
%To ensure the resulting merged data structure has all the relvant proporties of IntergratedSpaceSaving$\pm$ in the presence of new insertion and deletions, we need to update these counters.  More specifically, we increment each of the $m$ insert counts with the average of the insert count of all the discarded insert coutns, up to a max of $m$ counters.  The following Theorem demonstrates that the merged data structure preserves all 3 properties of Integrated SpaceSaving$\pm$, ie, the $\epsilon$ bound, the residual error bound and the relative error bound.

\begin{algorithm}[]
\SetAlgoLined
    \KwIn{Two Integrated SpaceSaving$\pm$ summaries of size $m$: $S_1$ and $S_2$}
    \KwOut{Merged summary $S_{merge}$ with size $m$}
    $S_{3} = S_1$ UNION $S_2$\;
    select $m$ largest items from $S_3$ based on insert counts\;
    add these $m$ items into $S_{merge}$ and keep their insert/delete counts\;
   
    % $remainCounts$ = $I_1 + I_2 - \sum_{i=1}^{m} insert_{i}^{merge}$\;
    % $w = remainCounts/m$\;
    % add $w$ to the insert counts of all items in $S_{merge}$\;

  \caption{Integrated SpaceSaving$\pm$ Merge Algorithm}
  \label{IntegratedSpaceSaving Merge}
\end{algorithm}

% f >= $\epsilon F$ keep. at most $insert_{minItem}^{1} + insert_{minItem}^{2}$ error.

% Randomly pick some number of items to make it has $m$ items.
% remaining item from S1 should not add more than $\epsilon F_{1}^{final} - minInsert^{1}$
% remaining item from S2 should not add more than $\epsilon F_{1}^{final} - minInsert^{2}$

% This is line of thinking might be hard to prove. Let's just add random items from the universe.

\begin{figure*}[h]
\centering
        \begin{subfigure}{.24\textwidth}
        \centering
        \includegraphics[scale = 0.6]{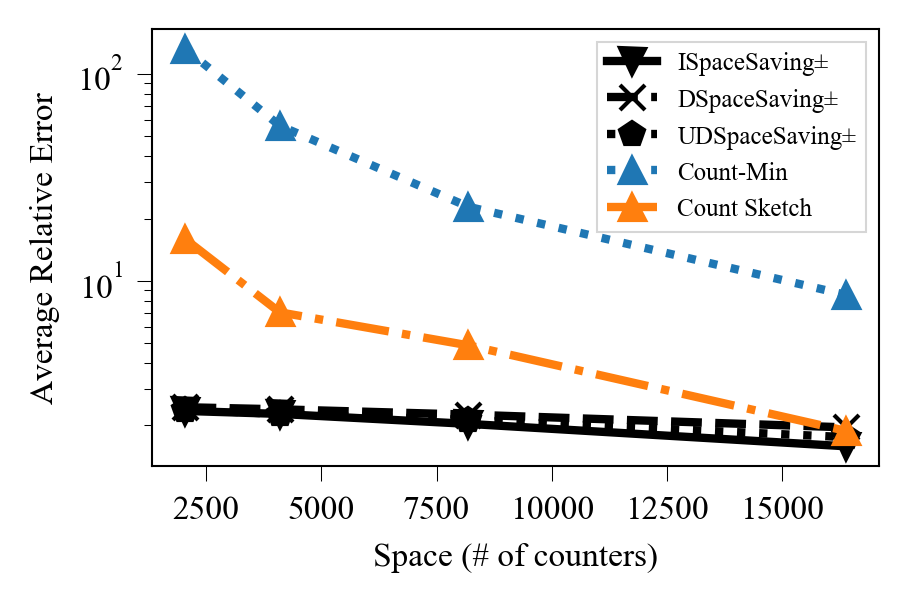
        }%
        \caption{Zipf $\beta=1.0$, Frequency Estimation}
        \end{subfigure}
        % \hfill
        \begin{subfigure}{.24\textwidth}
        \centering
        \includegraphics[scale = 0.6]{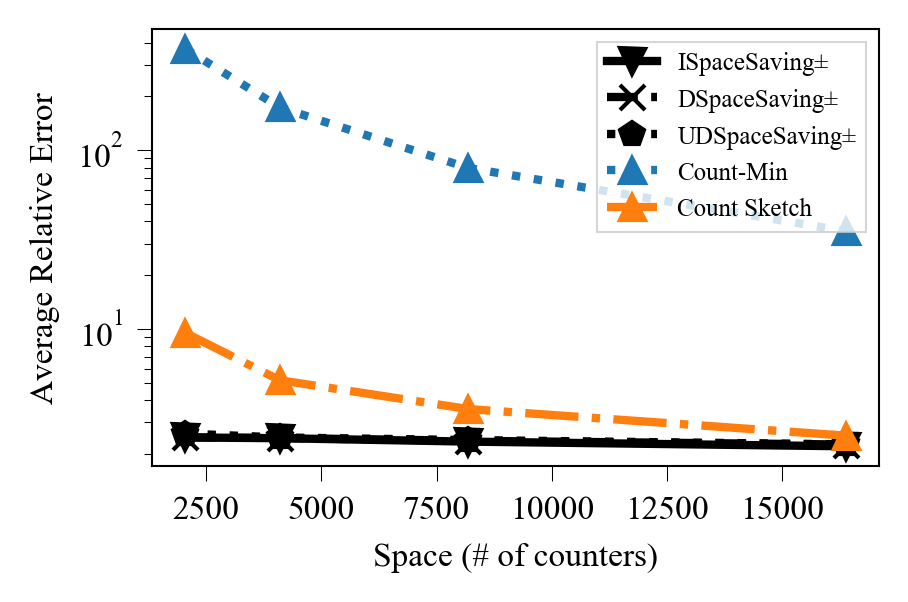}
        \caption{YCSB, Frequency Estimation}
        \end{subfigure}
        % \hfill
        \begin{subfigure}{.24\textwidth}
        \centering
        \includegraphics[scale = 0.6]{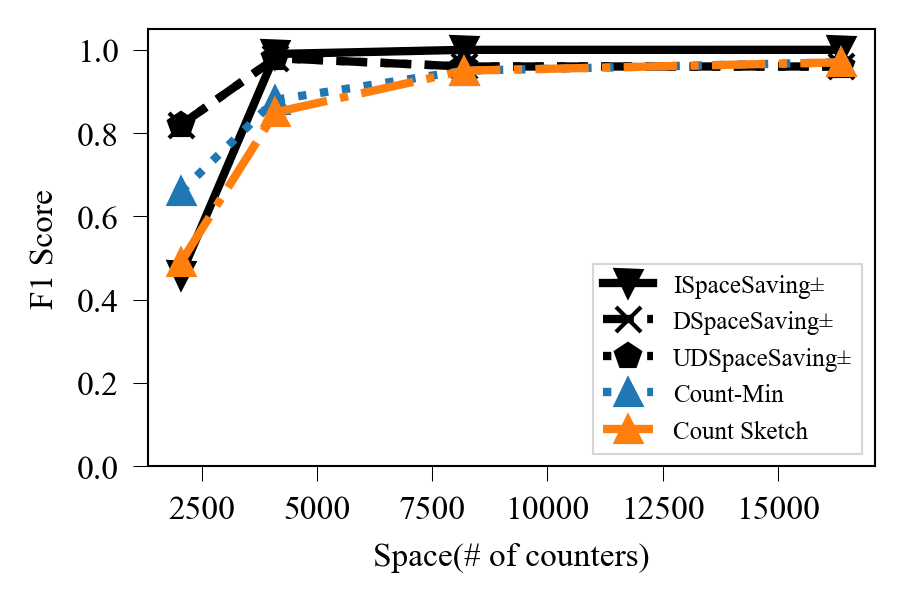}
        \caption{Zipf $\beta=1.0$, Top 100}
        \end{subfigure}
        % \hfill
         \begin{subfigure}{.24\textwidth}
        \centering
        \includegraphics[scale = 0.6]{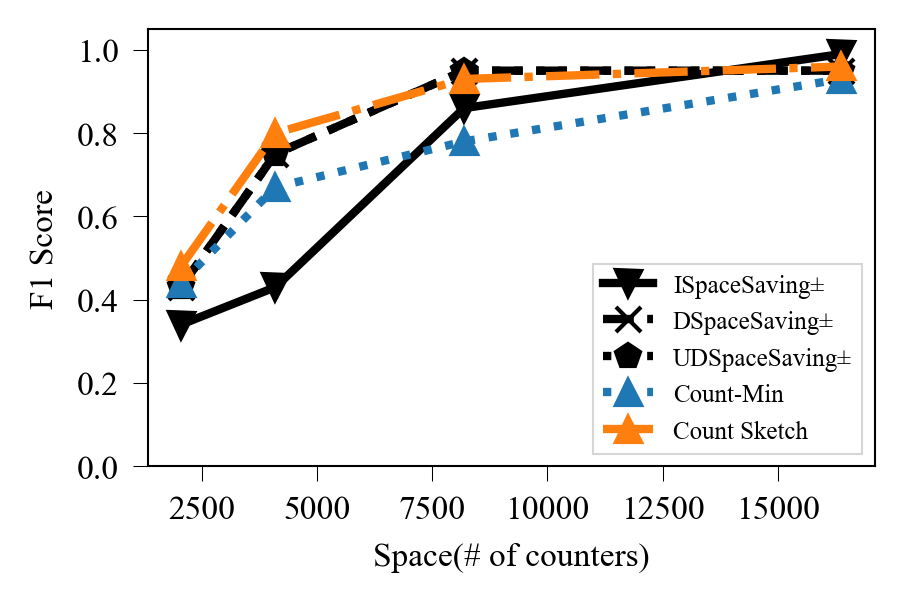}
        \caption{YCSB, Top 100}
        \end{subfigure}

	\caption{Comparison between our proposed algorithms with linear sketches over synthetic and real world dataset.}
\label{fig:exp}
\end{figure*}

\begin{theorem}
    The $IntegratedSpaceSaving\pm$ summaries are mergeable for $\alpha$ bounded deletion data stream following the Algorithm~\ref{IntegratedSpaceSaving Merge} with $m=\frac{\alpha}{\epsilon}$ counters.
\end{theorem}

\begin{proof}
    Assume the given $IntegratedSpaceSaving\pm$ summaries with $m=\frac{\alpha}{\epsilon}$ counters, $S_1$ and $S_2$, have processed two $\alpha$ bounded deletion data stream $\sigma^{1}$ and $\sigma^{2}$ with $I^{1}$ and $I^{2}$ insertions and $F_{1}^{1}$ and $F_{1}^{2}$ total frequencies respectively. Let $F_{1}^{final} = F_{1}^{1} + F_{1}^{2}$.
    First, we observe that The union step do not introduce any additional error. As a result, the error for these items is at most $\epsilon F_{1}^{1} + \epsilon F_{1}^{2} \leq \epsilon F_{1}^{final}$.
    %as the frequencies from $\sigma^{1}$ union $\sigma^{2}$. 

    % we must keep items with frequency larger than $\epsilon F_{1}^{final}$ in the new summary ($S_{merge}$) after union $S_1$ and $S_2$.
    % Next, we observe that evicting any remaining items cannot lead to severer frequency error, since their true frequency must be less than $\epsilon F_{1}^{final}$. However, 
    %We intent to keep the property of no underestimation for monitored item. As a result
    Since we don't want to underestimate any monitored items, Algorithm~\ref{IntegratedSpaceSaving Merge} keeps the largest $m$ items from the union-ed summary based on insert counts. We need to showcase that evicting all items except the largest $m$ items will not lead to error estimation larger than $\frac{I^{1} + I^{2}}{m}$, and we know that $\frac{I^{1} + I^{2}}{m} \leq \epsilon F_{1}^{final}$. This can be shown by contradiction. Assuming there exists an item $x$ in the union of $S_1$ and $S_2$ with true frequency larger or equal to $\frac{I^{1} + I^{2}}{m}$ and $x$ is not included than the $m$ largest items. Note, $x$'s insert count must be larger than its true frequency by Lemma~\ref{lemma: overestimation for IntegratedSpaceSaving}. Then this must imply the sum of the insert counts of $m$ largest item and $x$ is at least $(m+1)\cdot\frac{I^{1} + I^{2}}{m} > (I^{1}+I^{2})$. However, by Lemma~\ref{lemma: sum of insert counts equal to I}, we know the sum of insert counts cannot exceed $I^{1}+I^{2}$. As a result, evicting all items except the largest $m$ items from the union-ed summary based on insert counts does not lead to estimation error larger than $\epsilon F_{1}^{final}$. Therefore, Intergrated SpaceSaving$\pm$ is mergeable.

\end{proof} 

\section{Evaluation}
\label{sec-eval}
This section evaluates and compares the performance of our proposed algorithms with linear sketch from the turnstile model~\cite{li2014turnstile}. Our proposed counter based summaries have no assumption on the input universe but operates in the bounded deletion model. whereas linear sketch have dependency on the universe size but allow the entire data set to be deleted. The experiments aim to better understand the advantage and disadvantages of our proposed algorithms. The linear sketches that we compared with are:
\begin{itemize}
    \item \textbf{Count Min}~\cite{cormode2005improved}: Item's frequency is never underestimated.
    
    \item \textbf{Count Sketch}~\cite{charikar2002finding}: Provides an unbiased estimation, such that $E[\hat{f}(x)] = f(x)$ where $E$ is the expected value.
\end{itemize}

\textbf{Experiment Setup}. 
We implemented all of the algorithms in Python. Through all experiments, we set $\delta = U^{-1}$ for linear sketches to align the experiments with the theoretical literature~\cite{bhattacharyya2018optimal}.

\textbf{Data Sets}.
We use both synthetic and real world data consisting of items that are inserted and deleted. \textbf{Zipf Stream}: The insertions ($10^{5}$) are drawn from the Zipf distribution~\cite{zipf2016human} and deletions ($10^{5}/2$) are uniformly chosen from the insertions. Deletions happen after all the insertions. Cardinality is 22k.
\textbf{YCSB}: We use 60\% insertions and 40\% updates (delete followed by insert) with request distribution set to \textit{zipf} in YCSB  benchmark~\cite{cooper2010benchmarking} with interleaved 116645 insertions and 39825 deletions. Cardinality is 65k.

\textbf{Metrics}.
We use average relative error (ARE) as the metrics for frequency estimation. Let the final data set $D = \{x|f(x)>0\}$. ARE is calculated as $1/|D|(\sum_{x\in D} \frac{|f(x)-\claire{\hat{f}}(x)|}{f(x)})$. Lower ARE indicates more accurate approximations. For identifying the top 100 heavy hitters, we query all items in $D$ and report a set of 100 items with the largest frequency. Metrics is the F1 score.

% $\frac{2}{precision^{-1} + recall^{-1}}$. 
% With perfect precision and recall, F1 score is 1.

\subsection{Main Results}
Our proposed algorithms perform very well on frequency estimation task, as shown in Figure~\ref{fig:exp}(a)-(b). The y-axis is the average relative error and the x-axis the the number of counter used. For fair comparisons, we let IntergratedSpaceSaving$\pm$ uses three counters per entry, SpaceSaving and unbiased SpaceSaving use two counters, and linear sketch uses one counter per entry. We find that Integrated SpaceSaving$\pm$ is always the best, followed by Double SpaceSaving$\pm$ (DSS$\pm$) or Unbiased Double SpaceSaving$\pm$ ((UDSS$\pm$)), and Count Sketch always ranks at the fourth. For instance, when using 16384 counters over YCSB dataset, average relative error for IntergratedSpaceSaving$\pm$ is 2.20, DSS$\pm$ is 2.25, UDSS$\pm$ is 2.26, Count Sketch is 2.29, and Count-Min is 35.

In identifying the top 100 items, linear sketches are more competitive to our proposed algorithms. As shown in Figure~\ref{fig:exp}(c)-(d), the y-axis is the F1 score and x-axis is the number of counters. When more memory budget is allowed, Integrated SpaceSaving$\pm$ is always the best.
When the space budget up to 8192 counters, DSS$\pm$ and UDSS$\pm$ always outperform linear sketches. For instance, using 8192 counters over YCSB, DSS$\pm$ and UDSS$\pm$ both achieve 0.95 F1 score, whereas Count Sketch score 0.91 and Count Min scores 0.78.
\section{Conclusion}

This paper presents a detailed analysis of SpaceSaving$\pm$ family of algorithms with bounded deletion. We proposed new algorithms built on top of the original SpaceSaving and SpaceSaving$\pm$. They exhibit many desirable properties such as no underestimation, unbiased estimation, residual/ relative error guarantees and mergability. We believe these unique characteristics make our proposed algorithms a strong candidate for real-world applications and systems.

% \begin{acks}
% \end{acks}

\bibliographystyle{ACM-Reference-Format}
\bibliography{main-bib}


\end{document}